\documentclass[12pt]{article}
\usepackage{amssymb,amsmath,amsthm,amsfonts,mathrsfs}
\usepackage{enumerate}

\usepackage{pgf,tikz}
\usetikzlibrary{intersections,shapes}
\usepackage{graphicx}
\usepackage{float}

\usepackage[colorlinks=true, allcolors=blue]{hyperref}
\usepackage[all]{hypcap}
\usepackage[capitalize,nameinlink]{cleveref}

\usepackage{titlesec}
\titlelabel{\thetitle.\quad}

\newtheorem{theo}{Theorem}
\newtheorem{prop}[theo]{Proposition}
\newtheorem{lem}[theo]{Lemma}

\theoremstyle{definition}
\newtheorem{case}{Case}[theo]

\begin{document}
\title{On Hamiltonian-Connected and Mycielski graphs}
\def\correspondingauthor{\footnote{Corresponding author}}
\author{Ashok Kumar Das\correspondingauthor{}  and  Indrajit Paul\\Department of Pure Mathematics\\
University of Calcutta\\
35, Ballygunge Circular Road\\
Kolkata-700019\\
Email Address - ashokdas.cu@gmail.com \&
paulindrajit199822@gmail.com}
\maketitle
\begin{abstract}
A graph $G$ is Hamiltonian-connected if there exists a Hamiltonian path 
between any two vertices of $G$. It is known that if $G$ is 2-connected 
then the graph $G^2$ is Hamiltonian-connected. In this paper we prove that the 
square of every self-complementary graph of order grater than 4 is 
Hamiltonian-connected. If $G$ is a $k$-critical graph, then we prove that the Mycielski graph $\mu(G)$ is $(k+1)$-critical graph. Jarnicki et al.~\cite{jm} proved that for every 
Hamiltonian graph of odd order, the Mycielski graph $\mu(G)$ of $G$ is 
Hamiltonian-connected. They also pose a conjecture that if $G$ is 
Hamiltonian-connected and not $K_2$ then $\mu(G)$ is Hamiltonian-connected. In this paper we also prove this conjecture.
\end{abstract}
\section{INTRODUCTION}
A graph $G$ is Hamiltonian if it has a cycle containing all 
the vertices of $G$. Hamiltonian graphs have been extensively studied by 
several researchers. Ronald J. Gould has written a survey paper~\cite{g} on 
Hamiltonian graphs. But unfortunately, no easy testable characterization 
is known for Hamiltonian graphs. Bondy and Chv\'atal~\cite{bc} proved that a 
graph is Hamiltonian if and only if its closure $Cl(G)$ is Hamiltonian. 
\textit{Closure} of $G$ is formed by recursively joining two non-adjacent vertices 
of $G$ whose degree sum is at least $n$, where $n$ is the number of vertices 
of $G$. But this condition for Hamiltonicity doesn't help much as we are 
required to test another graph to be Hamiltonian. However, it 
provides a method for proving a sufficient condition for Hamiltoniancity. A condition that 
forces $Cl(G)$ to be Hamiltonian also forces $G$ to be Hamiltonian. For 
example, if $Cl(G)$ is complete then $G$ is also Hamiltonian. Chv\'atal~\cite{c} 
used this method to provide one of the strongest sufficient conditions 
for a graph to be Hamiltonian based on the degrees of the vertices of the 
graph.\par
A path in $G$ is \emph{Hamiltonian path} if it contains all 
the vertices of the graph $G$. A graph $G$ is \emph{Hamiltonian-connected} if for any two vertices $u,v\in V(G)$, there exists a $u$-$v$ 
Hamiltonian path. Obviously, a Hamiltonian-connected graph is 
Hamiltonian. However, the converse is not true.\par
Let $G=(V,E)$ be a connected graph and $u,v$ are two distinct vertices of 
$G$. Then \emph{distance} between $u$ and $v$, denoted by $d(u,v)$, is 
the length of the shortest path between $u$ and $v$. \emph{Diameter} of 
$G$, denoted by $diam(G)$ is the largest distance between the pair of vertices in 
$G$. The square of a graph $G$, denoted by $G^2$, is a 
graph with vertex set $V$ and $uv$ is an edge of $G^2$ if and only if  
$d(u,v)\leq2$ in $G$. Similarly, the cube of $G$, denoted by $G^3$, is the graph with vertex set $V$ and $uv$ is an edge of $G^3$ if and only if $d(u,v)\leq 3$. In this paper 
we give few more results on Hamiltonian-Connected graphs and Mycielski's 
graphs.\par
\section{Self-complementary Graphs and Hamiltonian Connectedness}
A graph is \emph{self-complementary} if the graph is isomorphic to its complement. A graph $G$ is \emph{k-connected} $(k\geq 1)$ if removal of vertices, (fewer than $k$) keeps the graph connected. If removal of a single vertex $v$ of $G$ disconnects the graph, then $v$ is a \emph{cut vertex} of $G$.\par
As mentioned in the introduction the following theorem, due to Chv\'atal~\cite{c}, is a sufficient condition for a graph to be Hamiltonian based on the degrees of the vertices of the graph.\par
\begin{theo}[\cite{c}]\label{t1}
Let $G$ be a graph of order $n>3$, the degrees $d_i$ of whose vertices satisfy $d_1\leq d_2\leq \cdots \leq d_n$. If for any $k<\frac{n}{2}$, we have $d_k\geq k+1\ or\ d_{n-k}\geq n-k$, then $G$ is Hamiltonian.
\end{theo}
It is easy to observe that a graph $G$ has a Hamiltonian path if and only if the graph $G\vee K_1$ has a Hamiltonian cycle, where $G\vee K_1$ is the graph obtained from $G$ by joining a new vertex $w$ to each vertex of $G$. Based on this observation, the next theorem analogous to the \cref{t1}, gives a sufficient condition for the existence of a Hamiltonian path in $G$.\par
\begin{theo}[\cite{we}]\label{t2}
Let $G$ be a graph of order $n\geq 2$, the degree $d_i$ of whose vertices satisfy $d_1\leq d_2\leq \cdots \leq d_n$. If for any $k< \frac{n+1}{2}$, we have $d_k\geq k\ or\ d_{n+1-k}\geq n-k$, then $G$ has a Hamiltonian path.
\end{theo}
Next, we state the theorem due to O.Ore, which provides a sufficient condition for a graph to be Hamiltonian-connected.\par
\begin{theo}[\cite{or}]\label{t3}
Let $G$ be a graph of order $n\geq 3$. If for any two non-adjacent vertices $u\ and\ v$ of $G,\ d(u)+d(v)\geq n+1,$ then $G$ is Hamiltonian-connected.
\end{theo}
It can be easily checked that the square of a tree of order at least 4 is 
not necessarily Hamiltonian. But M.Sekanina~\cite{se} and others 
proved that for any connected graph $G$, the \emph{cube} of $G$ is 
Hamiltonian-connected. Although  Nash-Williams and M.D. 
Plummer conjectured that for any 2-\emph{connected} graph $G$, the graph 
$G^2$ is Hamiltonian.  Fleischner~\cite{f} verified this conjecture. The 
work of Fleischner was strengthened by Chartrand, Hobbs, Jung, Kapoor and 
Nash-Williams\cite{ch} showing that the square of a 2-\emph{connected} graph is 
Hamiltonian-connected.\par
\begin{theo}[\cite{ch}]\label{t4}
Let $G$ be a 2-connected graph of order at least 3. Then $G^2$ is Hamiltonian-connected.
\end{theo}
A self complementary graph of order $n\geq 5$, may not be 2-connected. 
For example, the following self complementary graph is not 2-connected.
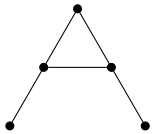
\begin{figure}[H]
\centering
\begin{tikzpicture}[scale=.90]
{\tikzstyle{every node}=[circle, draw, fill=black,
                        inner sep=0pt, minimum width=3pt]

\draw (0:0)node{}--++(60:1)node(a){}--++(60:1)node{}--++(-60:1)node(b){}--++(-60:1)node{} (a)--(b);
}
\end{tikzpicture}
\caption{Self-complementary graph on 5 vertices which is not 2-connected.}
\end{figure}
In this context we prove the following result.

\begin{theo}\label{t5}
If $G$ is self-complementary graph of order $n>3$, then $G$ has a Hamiltonian path. Moreover, if $n\geq 5$, then $G^2$ is Hamiltonian connected.
\end{theo}
\begin{proof}
The first part of the theorem is known \cite{glp}. But, for the sake of completeness  we give here a new proof. Let the degree sequence of $G$ be $d_1,d_2,...,d_k,...,d_n$ satisfying the condition $d_1\leq d_2\leq ...\leq d_k\leq...\leq d_n$, where $d_i$ is the degree of the vertex $v_i$. Then the degrees of $\overline{G}$ satisfy the condition \\
				\begin{equation}\label{e1}
				n-1-d_n\leq n-1-d_{n-1}\leq ...\leq n-1-d_k\leq ...\leq n-1-d_2\leq n-1-d_1
				\end{equation}
    Now, for any $k<\frac{n+1}{2}$, let $d_k\leq k-1$. Then 
    \begin{equation}\label{e2}
        n-1-d_k\geq n-k
    \end{equation}
    Again \cref{e1} can be written as \\
    $n-1-d_{(n+1)-1}\leq n-1-d_{(n+1)-2}\leq ...\leq n-1-d_{(n+1)-k}\leq ...\leq n-1-d_{(n+1)-(n-1)}\leq n-1-d_{(n+1)-n}$\\
    Because $G$ is self-complementary, $d_k=n-1-d_{n+1-k}$. So $d_{n+1-k}=n-1-d_k$. Therefore from \cref{e2}, $d_{n+1-k}\geq n-k$. Therefore from \cref{t2}, $G$ has a Hamiltonian path.
    
 We can also conclude that $G$ and hence $\overline{G}$ are connected.\par
To prove the second part, first we shall state the following known result.\par
\begin{lem}$\cite{ts}$\label{l1}
If $G$ is a self-complementary graph of order at least 5, then $diam(G)=2\ or\ 3$.
\end{lem}

Now we shall prove the second part of the \cref{t5}.\par \vspace{.3cm}
Let $G$ be a self complementary graph of order $n\geq 5$. We shall prove that $G^2$ is Hamiltonian-connected.\par

First, suppose $diam(G)=2$. Then $G^2$ is complete and hence $G^2$ is 
Hamiltonian-connected. Next, suppose $diam(G)=3$. If $G$ is 
2-\emph{connected} then by \cref{t4}, $G^2$ is Hamiltonian-connected. So, 
we assume $G$ has a cut-vertex $v$ and $G_1,G_2,\cdots,G_k$ are the 
components of $G\setminus\{v\}$. Let $G_i$ and $G_j$ be any two 
components of $G\setminus\{v\}$. Then every vertex of $G_i$ is adjacent 
to every vertex of $G_j$ in $\overline{G}$. Thus we conclude that 
$\overline{G}^2\setminus\{v\}=K_{n-1}$. Now $d_{G}(v)\geq 2$, then if 
$d_{\overline{G}}(v)\geq 3$, we consider two non-adjacent vertices of 
$\overline{G}^2$. One of these two vertices must be $v$ and, say, another 
vertex is $u$, then $d_{\overline{G}^2}(u) = n-2$. Since 
$d_{\overline{G}}(v)\geq 3$ implies $d_{\overline{G}^2}(v)\geq 3$, we 
have for non-adjacent vertices $u$ and $v$, $d_{\overline{G}^2}(u)+d_{\overline{G}^2}(v)\geq n-2+3=n+1$. Therefore by \cref{t3}, 
$\overline{G}^2$ is Hamiltonian-connected. Now, we examine the other possibilities.\par\vspace{.2cm}

\begin{case}\label{c1} Let $d_{\overline{G}}(v)=1$. Thus $d_{G}(v)=n-2$ and there exists only one 
vertex $x$ such that $vx\in E(\overline{G})$. Since 
$diam(\overline{G})=3$,  every vertex  in $\overline{G}$ is at a 
distance at most 2 from $x$. Now, if $d_{\overline{G}}(x)\geq 3$ then 
$d_{\overline{G}^2}(v)\geq 3$. Then, as before $\overline{G}^2$ is 
Hamiltonian-connected. Therefore assume $d_{\overline{G}}(x)=2$, where 
$N_{\overline{G}}(x)=\{v,z\}$ ( Figure 2). Since $diam(\overline{G})=3$, every 
vertex of $V(G)\setminus\{v\}$ is adjacent to $z\ in\ \overline{G}$. Thus 
$d_{\overline{G}}(z)=n-2$ and no other vertex in $\overline{G}$ is of 
degree $n-2$. Since $G$ is a self complementary graph and $z$ is a cut 
vertex of $\overline{G}$, as before ${G}^2\setminus\{z\}=K_{n-1}$. Again, 
all the vertices except $x$, that are adjacent to $z$ in $\overline{G}$ are also 
adjacent to $v$ in $G$. This implies $d_{{G}^2}(z)=n-2\geq 3$, as 
$n\geq 5$. So in ${G}^2$, $d_{{G}^2}(z)+d_{{G}^2}(u)\geq n-2+3 = n+1$, 
where $u$ is any other vertex of $G$ other than $z$. 
Hence $G^2$ is Hamiltonian-connected.
\end{case}

\begin{figure}[H]
\centering
\begin{tikzpicture}[scale=1.5]
{\tikzstyle{every node}=[circle, draw, fill=black,
                        inner sep=0pt, minimum width=3pt]

\draw (0:0)node(a)[label={[label distance=1pt]180:${v}$}]{}--++(0:1)node[label={[label distance=1pt]-90:${x}$}]{}--++ (0:1)node(b)[label={[label distance=1pt]-120:${z}$}]{}--++(40:1.2)node(c){} (b)--++(60:2)node(d){} (b)--++(-60:2)node(e){} ;
\draw[dashed] (a)to[out=80,in=160](d) (a)to[out=50,in=130](b) (a)to[out=-80,in=-160](e) (a)to[out=60,in=150](c);
}

\draw (b)--++(0:1.2)node(f)[ellipse, fill=white,inner sep=1pt, draw]{$\vdots$};
\draw[dashed]    (a)to[out=-60,in=-120](f);

\end{tikzpicture}
\caption{A possible figure in $\overline{G}$, where doted edges represent the edges of $G$.}
\label{f1}
\end{figure}
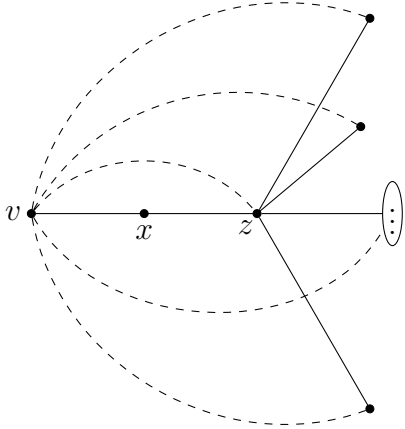

\begin{case}Let $d_{\overline{G}}(v)=2$. Since $v$ is a cut 
vertex of $G$, $v$ is not a cut vertex of $\overline{G}$. Now, assume 
$N_{\overline{G}}(v)=\{a,b\}$. Since $v$ is not a cut vertex of 
$\overline{G}$, at least one vertex of $a$ and $b$, say $a$, has a 
neighbour other than $b$ and $v$ in $\overline{G}$. This implies 
$d_{\overline{G}^2}(v)\geq 3$. Again, as before $\overline{G}^2\backslash v= K_{n-1}$. So for any two vertices $x, y$ of $G$ other than $v$, we have $d(x)+d(y)= 2n-4$ in $K_{n-1}$. As $n\geq 5, 2n-4\geq n+1$, i.e. $d_{\overline{G}^2}(x) + d_{\overline{G}^2}(y)$$\geq n+1$. Finally,  $d_{\overline{G}^2}(v) + d_{\overline{G}^2}(x)$$\geq n-2+3= n+1$, where $x\in V-\{v\}$. Hence $\overline{G}^2$ is Hamiltonian-connected and accordingly $G^2$ is 
Hamiltonian-connected. This completes the proof of the \cref{t5}.\qedhere
\end{case}
\end{proof}

\section{Mycielski's Graphs}
Let $G=(V,E)$ be a connected graph with vertex set $V=\{ v_1,v_2,\cdots ,v_n \}$. The Mycielski's graph $\mu (G)$ of $G$ is a graph with the vertex set $X\cup Y \cup \{z\}$, where $X=\{ x_1,x_2,\cdots , x_n \}$ , $Y=\{ y_1,y_2,\cdots , y_n \}$ , and the vertices of $X$ induce $G$. The new edges in $\mu (G)$ are $zy_i$ for all $i$ also  $x_iy_j$ is an edge of $\mu (G)$ if $x_ix_j$ is an edge of $G$.
For example, $\mu (K_2)=C_5$ and $\mu (C_5)$ is the Gr\"{o}tzsch graph(Figure 3).

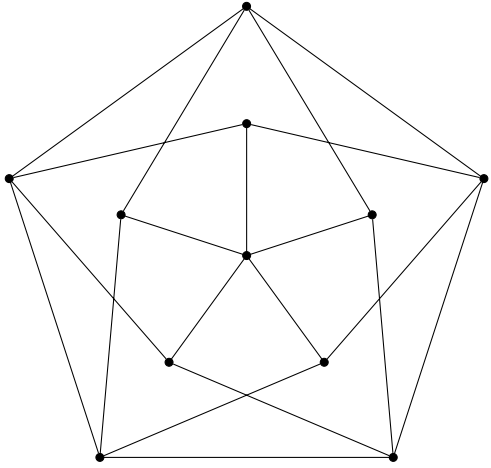
\begin{figure}[H]
\centering
\begin{tikzpicture}[scale=3.9]
{\tikzstyle{every node}=[circle, draw, fill=black,
                        inner sep=0pt, minimum width=3pt]
\draw (0:0) node(a) {} --++ (0:1) node(b) {}--++(-108:-1)node(c){}--++(144:1)node(d){}--++(36:-1)node(e){}--(a);
\path [name path=a1] (a) --++ (54:1.5);
\path [name path=d1] (d) --++ (-90:1.5);
\path [name intersections={of=a1 and d1, by=f}];
\draw (f)node{}--++(90:.45)node(g){} (f)--++(162:.45)node(h){} (f)--++(234:.45)node(i){} (f)--++(306:.45)node(j){} (f)--++(378:.45)node(k){};
\draw (h)--(a)--(j) (i)--(b)--(k) (j)--(c)--(g) (k)--(d)--(h) (g)--(e)--(i);
}
\end{tikzpicture}
\caption{The Gr\"{o}tzsch Graph.}
\end{figure}

 Thus starting with $M_2=K_2$, the Mycielski's graph $M_n$ are formed by iterating $\mu $ successively $(n-2)$ times on $K_2$. Thus $M_3=\mu (K_2)$,  $M_4=\mu (M_3)$ and in general $M_{k+1}=\mu (M_k)$. \\Mycielski's  graphs are of interest because if $G$ is triangle-free then  $\mu (G)$ is also triangle-free and if the chromatic number $\chi (G)$ of $G$ is $k$ , then $\chi (\mu (G))=k+1$.\par
Now, $M_3$ is $C_5$; also it can be observed that $M_4$, the Gr\"{o}tzsch graph is Hamiltonian. Fisher et al.~\cite{fm} proved that if $G$ is Hamiltonian then $\mu (G)$ is also Hamiltonian. Jarnicki et al.~\cite{jm} have extended this  result and proved that $\mu (G)$ is Hamiltonian-connected if $G$ is Hamiltonian and of odd order. In this paper we prove that each Mycielski's graph is a $k$-critical graph. Finally, we shall prove a conjecture posed by  Jarnicki et al.~\cite{jm} regarding Hamiltonian-connectedness of Mycielski's graphs.
First, we shall prove the following proposition.
\begin{prop}
If $P_n$ is a path on $n$ vertices, then $\mu (P_n)$ has a Hamiltonian path.
\end{prop}
\begin{proof}
    Let $x_1 x_2 \ldots x_{n-1}x_n$ be the path $P_n$. First, suppose $n$ is odd, then $\mu (P_n)$ has a Hamiltonian path, 
     $x_1\, y_2\, x_3\, y_4 \ldots  x_{n-2}\, y_{n-1}\, x_n\, x_{n-1}\, y_n\, z\, y_1\, x_2\, y_3\, x_4 \ldots x_{n-3}\, y_{n-2}$.
    Next, if $n$ is even then 
    $x_1\, y_2\, x_3\, y_4 \ldots x_{n-1}\, y_n\, z\, y_1\, x_2\, y_3\, x_4 \ldots y_{n-1}\, x_n$
    is a Hamiltonian path in $\mu (P_n)$.
\end{proof}
A graph $G$ is \textbf{$k$-critical graph} if $\chi (G)=k$  but $\chi (H)<\chi (G)$ for every proper subgraph $H$ of $G$. Now we prove the following propositions.\par
\begin{prop}
Let $G$
 be a $k$-chromatic graph, where $k\geq 2$. Then for every vertex $v$ of $G$, $\chi (G-v)=k $ or $k-1$.
\end{prop}
\begin{proof}
On the contrary, Let $\chi (G-v)=p \leq k-2 $. Now, if $v$ is adjacent to the vertices in $G-v$ which receives all the $p$ colors. Then we can color the vertex $v$ by $(p+1)$th color. Which implies that $\chi (G) \leq k-1$. Thus we arrive at a contradiction. Hence, $\chi (G-v)=k $ or $k-1$.
\end{proof}
\noindent\textbf{Proposition 9.} Let $G$ be a $k$-chromatic graph, where $k\geq2$. Then for every edge $e$ of $G$, $\chi(G-e)=k$ or $k-1$.\vspace{.3cm}\\
\textit{Proof.} On the contrary, assume $\chi(G-e)=p\leq k-2$ where $e=x_ix_j$. Now in $G-e$ if both $x_i$ and $x_j$ receive the $p$th color then recolor $x_i$ or $x_j$ by $(p+1)$th color. Which is a contradiction. Therefore $\chi(G-e)=k$ or $k-1$.\qed

From the above two propositions we conclude that if $G$ is $k$-critical graph, then for every vertex $v$, $\chi (G-v)=k-1$ and for every edge $e$, $\chi((G-e))=k-1$.
Now, we prove the following theorem.\par

\noindent\textbf{Theorem 10.} \textit{Every Mycielski graph $M_k$ is a $k$-critical graph for $k \geq 2$.}\\

\textit{Proof.}To prove this theorem, we use mathematical induction on $k$. We can easily see that $M_2=K_2$ and $M_3=C_3$ are respectively $2$-critical and $3$-critical graphs. Now, we assume that $M_p$ is a $p$-critical graph for $p>3$ and we shall prove that $M_{p+1}$ is $(p+1)$-critical graph.\\
Again $\chi (M_p)=p$ and suppose $V(M_p)=\{x_1,x_2,\cdots , x_n\}$ . The vertex set  $V(M_{p+1})$, of $M_{p+1}$ is $\{x_1,x_2,\cdots , x_n\}\cup \{y_1,y_2,\cdots , y_n\}\cup \{z\}$. We know that $\chi (M_{p+1})=p+1$ and we shall prove that $\chi (M_{p+1}-v)=p$ for $v \in V(M_{p+1})$. We consider the following cases.\par
\noindent\textbf{Case 1.} Let $v = z$. Since $\chi (M_p)=p$, from the construction of $M_{p+1}$, it follows that $\chi (M_{p+1}-z)=p$.\par
\noindent\textbf{Case 2.} Let $v=x_i \in V(M_p)$. Since $M_p$ is $p$-critical, $\chi (M_{p}-x_i)=p-1$. Suppose $1,2,3,\cdots ,(p-1)$ colors are needed to color the vertices $x_1,x_2,\cdots ,x_{i-1},$ $x_{i+1},\cdots x_n$.\par
Again, each vertex $y_i$ where $(1\leq i\leq n)$ is adjacent to some vertices of the set $\{ x_1, x_2, \ldots ,\\ x_{i-1}, x_{i+1}, \ldots , x_n \}$. So we can use $p$-th color to color the vertices $\{y_1,y_2,\cdots , y_n\}$. Then we can use any of the colors $1,2,3,\cdots ,(p-1)$ to color the vertex $z$. Thus $\chi (M_{p+1}-x_i)=p$.\par
\noindent\textbf{Case 3.} Let $v=y_i$. Now $x_i$ and $y_i$ are non adjacent and $x_i$ and $y_i$ have same adjacency in $V(M_{p+1})\backslash \{z\}$ for each $i\ (1\leq i\leq n)$. Since $M_p$ is $p$-critical,  $\chi (M_{p}-x_i)=p-1$ and we recolor the vertices $x_1,x_2,\cdots ,x_{i-1},x_{i+1},\cdots , x_n$ by $1,2,3,\cdots ,(p-1)$ colors. We can color the vertices $y_j$ and $x_j$ by same color for every $j \ne i$. Then we color the vertices $x_i$ and $z$ by $p$-th color. Thus  $\chi (M_{p+1}-y_i)=p$.\\
Next, we consider the edge deletion cases.\par
\noindent\textbf{Case 4.} Let $e=x_ix_j$. Since $M_p$ is $p$-critical, $\chi(M_p-e)=p-1$. Again, existance of the edge $x_ix_j$ implies the existance of the edges $x_iy_j$ and $x_jy_i$ in $M_{p+1}$. Now the vertices of $M_p-e$ can be colored by $(p-1)$ numbers of colors, say $1,2,...,(p-1)$; where $x_i$ and $x_j$ both receive the same color. Now in $M_{p+1}$ the vertices $y_1, y_2, ..., y_i, ..., y_j, ..., y_n$ are colored by $1, 2, ..., p-1,p$ colors (i.e. $p$ number of colors). Next, we delet the edges $x_iy_j$ and $x_jy_i$ from $M_{p+1}-e$. Now, since in $M_{p+1}-\{x_iy_j, x_jy_i, x_ix_j\}$ the vertices $x_k$ and $y_k$ $(1\leq k\leq n)$ have the same neighbours in the vertex set of $M_p-x_ix_j$, the vertices $y_1,y_2,...,y_n$ now can be colored by $1,2,...,(p-1)$ colors. Next, we recolor the vertices $x_i$ and $x_j$ in $M_p-e$ by the $p$th color and insert the edges $x_iy_j$ and $x_jy_i$ to obtain the graph $M_{p+1}-e$. Then the $p$th color can be used to color $z$. This proves that $\chi(M_{p+1}-e)=p$.\par
\noindent\textbf{Case 5.} Let $e=x_iy_j$. Now, color the vertices $x_1, x_2, ..., x_i, ..., x_{j-1}, x_{j+1}, ..., x_n$ of $M_{p+1}$ by $(p-1)$ colors as above and use the $p$th color to color the vertex $x_j$. Again in $M_{p+1}-e$, the edge $x_jy_i$ exists and the vertex $x_j$ is colored by $p$th color. So the vertices $y_1, y_2, ..., y_n$ can be colored by $1, 2, ..., (p-1)$ colors in $M_{p+1}-e$ as in the case $4$. Thus we can use the $p$th color to color the vertex $z$. Hence $\chi(M_{p+1}-e)=p$.\par
\noindent\textbf{Case 6.} Let $e=y_iz$. Recolor the vertices of $M_p$ by $1,2, ..., p$ colors, where $x_i$ is colored by $p$th color as $M_p-x_ix_j$ is colored by $1,2,...,(p-1)$ colors. In $M_{p+1}$ the vertices $x_k$ and $y_k (1\leq k\leq n)$ receive the same color as their adjacencies are same in $M_p$. Thus the vertex $y_i$ can also be colored by $p$th color. So in $M_{p+1}-e$, the vertex $z$ can be colored by $p$th color. Hence $\chi(M_{p+1}-e)=p$.
\par Hence $M_{p+1}$ is $(p+1)$-critical graph. This completes the proof of this theorem.\qed\\

From the proof of Theorem 10, it follows that the above result can be generalized for arbitrary $k$-critical graph $G$.

\noindent\textbf{Corollary 11.} \textit{If $G$ is any $k$-critical graph, then $\mu(G)$ is $(k+1)$-critical graph.}\vspace{.3cm}

We have mentioned earlier that Jarnickie et al.~\cite{jm} proved that if $G$ is Hamiltonian and of odd order then $\mu (G)$ is Hamiltonian-connected. They have also proved that this result does not hold if the order of $G$ is even. In this context they have conjectured the following, which we shall prove in the next theorem. \vspace{.2cm}\\
\noindent\textbf{Theorem 12.}\textit{ If the graph $G$ is Hamiltonian-connected and not $K_2$, then $\mu (G)$ is also Hamiltonian-connected.}\vspace{.3cm}\\
To prove the above theorem, first we prove the following lemmas. \vspace{.1cm}\\

\noindent \textbf{Lemma 13.} \textit{If $G$ is a Hamiltonian-connected graph, then $G$ must 
contains odd cycle.}

\noindent\textit{Proof.} If $G$ is bipartite then $G$ is not Hamiltonian-connected. Thus $G$ must contains an odd cycle.\vspace{.01cm}\\

\noindent\textbf{Lemma 14.}\textit{If $G$ is a Hamiltonian-connected graph of order $\geq4$, then degree of $v$, $d(v)\geq 3$ for all $v \in V(G)$.}

\begin{proof}
Let $n$ be the order of $G$. Suppose $v_1\ v_2\ v_3\ldots v_n\ v_1$ is a 
Hamiltonian cycle in $G$. Next, assume 
$d(v_2)=2$. Then we do not have any $v_1$-$v_3$ 
Hamiltonian path in $G$. Which contradicts that $G$ is 
Hamiltonian-connected and complete the proof of the lemma.
\end{proof}

Now we prove the Theorem 12.
\begin{proof}
Let $G=(V,E)$ be a Hamiltonian-connected graph, where $V = \{ v_1, v_2, v_3, \ldots , v_{n-1},\\ v_n \}$. Since Hamiltonian-connected graph must be Hamiltonian, so if $n$ is odd then by Theorem 5.1 of \cite{jm},  $\mu (G)$ is Hamiltonian-connected. Thus we suppose that $n$ is even.\par
Now, we assume $\mu (G)$ has vertices $X \cup Y \cup \{z\}$ where $X=\{x_1, x_2, \ldots , x_n\}$ , $Y=\{y_1, y_2, \ldots , y_n\}$. Also $G[X]$, the subgraph induced by the vertex set $X$ is $G$ and $x_1, x_2, x_3, \ldots ,\\x_{n-1}, x_n, x_1$ is a Hamiltonian cycle. Also, $\mu (G)$ has extra edges $zy_i$ for all $i$, and $y_ix_j$ and $y_jx_i$ are edges of $\mu (G)$ if $x_ix_j \in E(G)$. To prove that $\mu (G)$ is Hamiltonian-connected, we have to show that for any two vertices $u,v \in V(\mu (G))$ there exists a $u$-$v$ Hamiltonian path in $\mu (G)$.\par
Now, without loss of generality we consider the following cases.\par
\noindent\textbf{Case 1.} Let $x_1, x_k \in V(G)$. Since $G$ is Hamiltonian-connected suppose we have the following  $x_1$-$x_k$ Hamiltonian path: $x_1 \  x_j \  x_l \cdots x_i \  x_m \hspace{5pt} x_k$ in $G$. Then we construct the following $x_1$-$x_k$ Hamiltonian path $P_1$ in $\mu (G)$ where 
$$P_1: x_1 \  y_j \  x_l \cdots y_i \  x_m \  y_k \  z \  y_1 \  x_j \  y_l \cdots x_i \  y_m \  x_k$$
\textbf{Case 2.} Let $x_1, y_p \in V(\mu (G))$. Suppose we have the following $x_1$-$x_p$ Hamiltonian path $P$ in $G$, where
$$P: x_1 \  x_i \  x_j \  x_k \  x_r \  x_s \ldots x_t\ x_l \  x_m \  x_p$$
Now, we construct the following Hamiltonian $x_1$-$y_p$ path $P_1$ in $\mu (G)$ where 
$$P_1: x_1 \  y_i \  x_j \  y_k \ldots \  x_t \  y_l\ x_m \  x_p \  y_m \  x_l \  y_t \ldots x_k \  y_j \  x_i\ y_1\ z\ y_p$$
\textbf{Case 3.} Let $x_1,z \in V(\mu(G))$. Since $G$ is Hamiltonian-connected, assume $P'$ be the $x_1$-$x_n$ Hamiltonian path in $G$, with $x_1 x_n \in E(G)$, where $P': x_1x_2...x_ix_r...x_qx_s...x_n$. Now, $n$ is even and $G$ contains an odd cycle, also $d(v)\geq 3$ for all $v\in V(G)$, there must exist two even(or odd) indexed vertices of $P'$ that are adjacent. If $x_r$ and $x_s$ are two vertices of odd index, then we relable the vertices of the Hamiltonian cycle $x_1x_2...x_ix_r...x_qx_s...x_nx_1$ so that $x_r$ and $x_s$ are of even index. Thus without loss of generality, let $x_r$ and $x_s$ are two such vertices of even index. Next, as $G$ is Hamiltonian-connected, there exists a $x_i$-$x_r$ Hamiltonian path $P''$ in $G$, where $P'': x_ix_j ... x_s ... x_tx_r$, and $x_ix_r$ and $x_sx_r$ $\in E(G)$.\\
Now, we relable the vertices of $P''$ to get $P$, where $P: x_1x_2x_3...x_{p-1}x_px_{p+1}...x_{n-1}x_n$ and $x_1x_n, x_px_n\in E(G)$ (i.e. we relable $x_i$ by $x_1$, $x_j$ by $x_2$,..., $x_s$ by $x_p$, $x_r$ by $x_n$ etc.).\\
Next, we construct the $x_1$-$z$ Hamiltonian path $P_1$ in $\mu(G)$ as follows (Fig. 4). \\$$P_1: x_1y_2x_3y_4x_5...y_{p-2}x_{p-1}y_px_{p+1}y_{p+2}...x_{n-3}y_{n-2}x_{n-1}y_nx_py_{p-1}x_{p-2}...y_5x_4y_3x_2y_1x_ny_{n-1}x_{n-2}$$ \rightline{$y_{n-3}...x_{p+2}y_{p+1}z$} \\
\begin{figure}[H]
\tikzstyle{every node}=[circle, draw, fill=black,
                        inner sep=0pt, minimum width=3pt]
\centering
\begin{tikzpicture}[scale=1]

\draw (0:0) node [label={[label distance=1pt]90:${\scriptstyle x_1}$}](x_1){} ++ (0:1)node [label={[label distance=1pt]90:${\scriptstyle x_2}$}](x_2){} ++ (0:1) node [label={[label distance=1pt]90:${\scriptstyle x_3}$}](x_3){} ++(0:1) node [label={[label distance=1pt]90:${\scriptstyle x_4}$}](x_4){}++ (0:1)node [label={[label distance=1pt]90:${\scriptstyle x_5}$}](x_5){} ++(0:2) node [label={[label distance=1pt]90:${\scriptstyle x_{p-2}}$}](x_{p-2}){} ++(0:1) node [label={[label distance=1pt]90:${\scriptstyle x_{p-1}}$}](x_{p-1}){} ++(0:1) node [label={[label distance=1pt]90:${\scriptstyle x_p}$}](x_p){} ++(0:1) node [label={[label distance=1pt]90:${\scriptstyle x_{p+1}}$}](x_{p+1}){} ++(0:1) node [label={[label distance=1pt]90:${\scriptstyle x_{p+2}}$}](x_{p+2}){} ++(0:2) node [label={[label distance=1pt]90:${\scriptstyle x_{n-3}}$}](x_{n-3}){} ++(0:1) node [label={[label distance=1pt]90:${\scriptstyle x_{n-2}}$}](x_{n-2}){} ++(0:1) node [label={[label distance=1pt]90:${\scriptstyle x_{n-1}}$}](x_{n-1}){} ++(0:1) node [label={[label distance=1pt]90:${\scriptstyle x_n}$}](x_n){} ;
\draw (x_1)++(-90:5) node [label={[label distance=1pt]-90:${\scriptstyle y_1}$}](y_1){} ++(0:1) node [label={[label distance=1pt]-90:${\scriptstyle y_2}$}](y_2){} ++(0:1) node [label={[label distance=1pt]-90:${\scriptstyle y_3}$}](y_3){} ++(0:1) node [label={[label distance=1pt]-90:${\scriptstyle y_4}$}](y_4){} ++(0:1) node [label={[label distance=1pt]-90:${\scriptstyle y_5}$}](y_5){}++(0:2) node [label={[label distance=1pt]-90:${\scriptstyle y_{p-2}}$}](y_{p-2}){} ++(0:1) node [label={[label distance=1pt]-90:${\scriptstyle y_{p-1}}$}](y_{p-1}){} ++(0:1) node [label={[label distance=1pt]-90:${\scriptstyle y_p}$}](y_p){} ++(0:1) node [label={[label distance=1pt]-80:${\scriptstyle y_{p+1}}$}](y_{p+1}){} ++(0:1) node [label={[label distance=1pt]-90:${\scriptstyle y_{p+2}}$}](y_{p+2}){} ++(0:2) node [label={[label distance=1pt]-90:${\scriptstyle y_{n-3}}$}](y_{n-3}){} ++(0:1) node [label={[label distance=1pt]-90:${\scriptstyle y_{n-2}}$}](y_{n-2}){} ++(0:1) node [label={[label distance=1pt]-90:${\scriptstyle y_{n-1}}$}](y_{n-1}){} ++(0:1) node [label={[label distance=1pt]-90:${\scriptstyle y_n}$}](y_n){};

\draw(y_p)++(-90:4) node [label={[label distance=1pt]-90:${\scriptstyle z}$}](z){} ;
\draw[->,very thick] (x_1)--(y_2);
\draw[->,very thick] (y_2)--(x_3);
\draw[->,very thick] (x_3)--(y_4);
\draw[->,very thick] (y_4)--(x_5);
\draw[->,very thick][dashed] (x_5)--(y_{p-2});
\draw[->,very thick](y_{p-2})--(x_{p-1});
\draw[->,very thick](x_{p-1})--(y_p);
\draw[->,very thick] (y_p)--(x_{p+1});
\draw[->,very thick] (x_{p+1})--(y_{p+2});
\draw[->,very thick][dashed] (y_{p+2})--(x_{n-3});
\draw[->,very thick] (x_{n-3})--(y_{n-2});
\draw[->,very thick] (y_{n-2})--(x_{n-1});
\draw[->,very thick] (x_{n-1})--(y_n);
\draw[->,very thick] (y_n)--(x_p);
\draw[->,very thick] (x_p)--(y_{p-1});
\draw[->,very thick] (y_{p-1})--(x_{p-2});
\draw[->,very thick] [dashed](x_{p-2})--(y_5);
\draw[->,very thick] (y_5)--(x_4);
\draw[->,very thick](x_4)--(y_3);
\draw[->,very thick] (y_3)--(x_2);
\draw[->,very thick] (x_2)--(y_1);
\draw[->,very thick] (y_1)--(x_n);
\draw[->,very thick] (x_n)--(y_{n-1});
\draw[->,very thick] (y_{n-1})--(x_{n-2});
\draw[->,very thick] (x_{n-2})--(y_{n-3});
\draw[->,very thick][dashed] (y_{n-3})--(x_{p+2});
\draw[->,very thick] (x_{p+2})--(y_{p+1});
\draw[->,very thick] (y_{p+1})--(z);

\end{tikzpicture}
\caption{The $x_1$-$z$ path $P_1$.}
\end{figure}
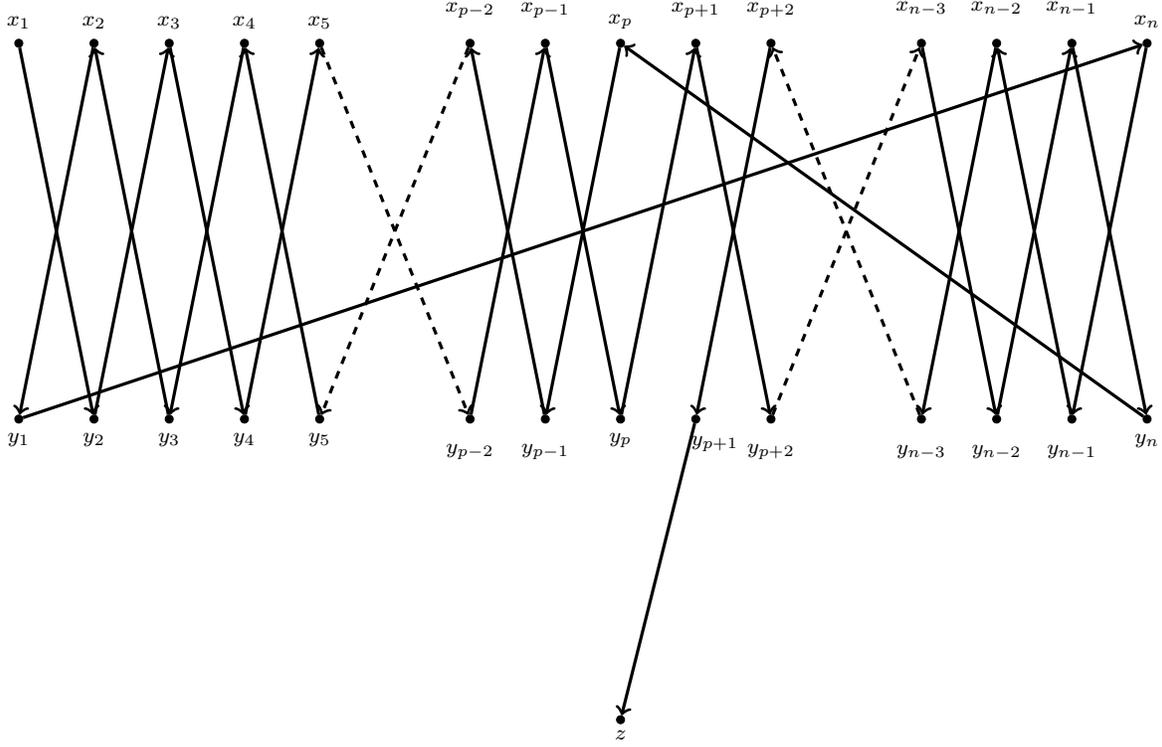
\textbf{Case 4.} Let $x_1, y_1\in V(\mu(G))$. Now we have the following $x_1$-$x_n$ Hamiltonian path $P$ in $G$ where $P: x_1x_2x_3x_4...x_{r-1}x_rx_{r+1}...x_{n-2}x_{n-1}x_n$. Again, as in the case 3, suppose $x_1x_n$ and $x_rx_n$ $\in E(G)$ and $x_r$ is a vertex of even index. Then we construct the following $x_1$-$y_1$ Hamiltonian path $P_1$ in $\mu(G)$, where\\ $$P_1: x_1x_2y_3x_4y_5...y_{r-1}x_ry_nx_{n-1}y_{n-2}... y_{r+2}x_{r+1}y_rx_{r-1}... y_4x_3y_2zy_{r+1}x_{r+2}...x_{n-2}y_{n-1}x_ny_1$$
\\\textbf{Case 5.} Let $y_1, z \in V(\mu(G))$ also $y_1 z$ is an edge of $\mu (G)$. Let $x_1 x_l \in E(G)$, then since $G$ is Hamiltonian-connected we have a $x_1$-$x_l$ Hamiltonian path $P$ in $G$ where 
$$P: x_1 \  x_i \  x_j \  x_m \cdots x_k \  x_p \  x_l$$
Thus in $\mu (G)$, we have the follwing $y_1$-$z$ Hamiltonian path $P_1$ where
$$P_1: y_1 \  x_i \  y_j \  x_m \cdots x_k \  y_p \  x_l \  x_1 \  y_i \  x_j \cdots y_k \  x_p \  y_l \  z $$
\textbf{Case 6.} Let $y_1,y_2 \in V(\mu (G))$ where $x_1x_2\in E(G)$. Now, let $x_1 x_l \in E(G)$ and assume the following $x_1$-$x_l$ Hamiltonian path $P$ in $G$, where $P: x_1 \  x_2 \  x_i \  x_j \cdots x_k \  x_m \  x_l$. Then we construct the following $y_1$-$y_2$ Hamiltonian path $P_1$ in $\mu (G)$, where
$$P_1: y_1 \  x_2 \  x_i \  y_j \cdots y_k \  x_m \  y_l \  x_1 \  x_l \  y_m \  x_k \cdots x_j \  y_i \  z \  y_2$$
\textbf{Case 6.1.} Let $y_1, y_p\in V(\mu (G))$, where $p$ is even and $x_1 x_p \notin E(G)$. Also, assume $x_1x_n \in E(G)$. Since $G$ is Hamiltonian-connected, $P$ be the $x_1$-$x_n$ Hamiltonian path, where  
$P: x_1 \  x_2 \  x_3 \  x_4 \  x_5  \ldots x_{p-1}x_px_{p+1}\ldots x_{n-2} \  x_{n-1} \  x_n$. Next, we construct the $y_1$-$y_p$ Hamiltonian path $P_1$ in $\mu (G)$ as follows ( Figure 5).
$$P_1: y_1 \  x_2 \  y_3 \  x_4 \  y_5 \cdots y_{p-1}x_px_{p+1}y_{p+2}\cdots y_{n-2} \  x_{n-1}\ y_n  \  x_1\ x_n \   y_{n-1} \  x_{n-2}$$ \rightline{$\cdots x_{p+2}\ y_{p+1}  \  z \  y_2 \  x_3 \  y_4\cdots x_{p-1}\ y_p$}.\\
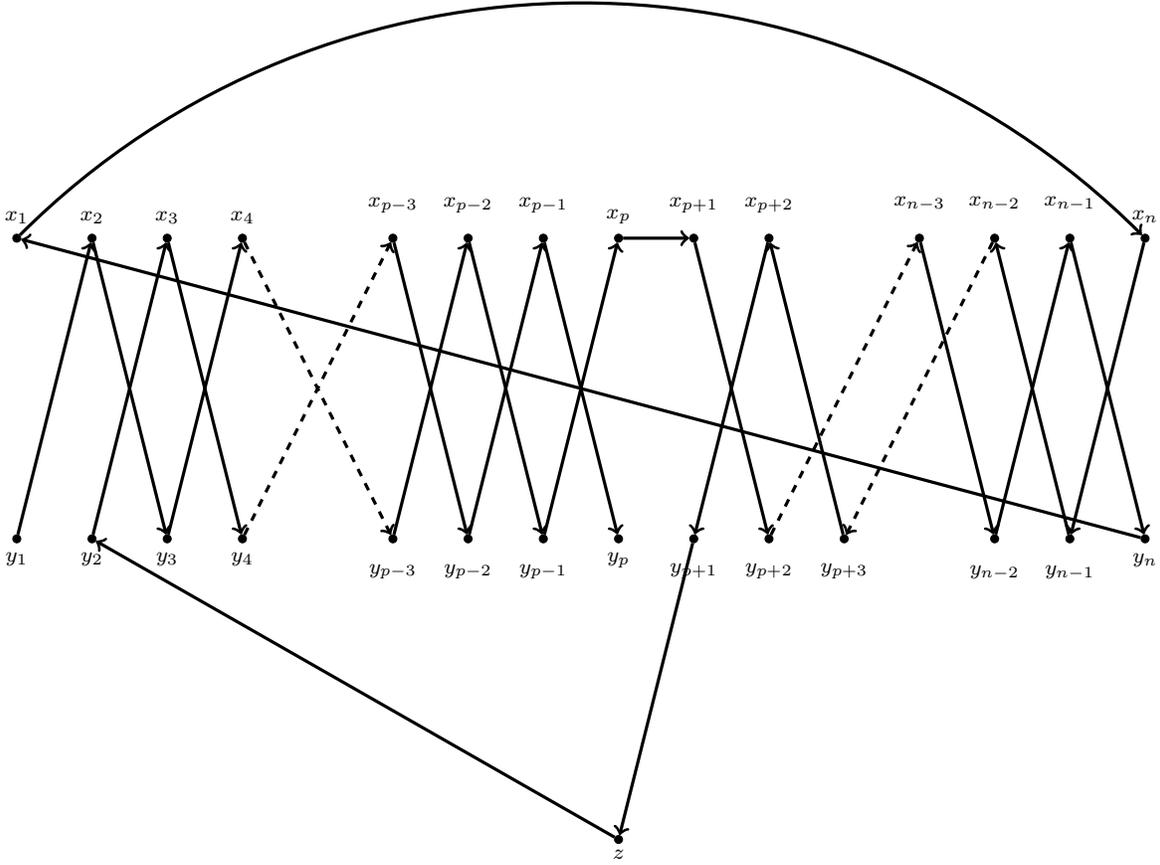
\begin{figure}[H]
\tikzstyle{every node}=[circle, draw, fill=black,
                        inner sep=0pt, minimum width=3pt]
\centering
\begin{tikzpicture}[scale=1]

\draw (0:0) node [label={[label distance=1pt]90:${\scriptstyle x_1}$}](x_1){} ++ (0:1)node [label={[label distance=1pt]90:${\scriptstyle x_2}$}](x_2){} ++ (0:1) node [label={[label distance=1pt]90:${\scriptstyle x_3}$}](x_3){} ++(0:1) node [label={[label distance=1pt]90:${\scriptstyle x_4}$}](x_4){}++ (0:2) node [label={[label distance=1pt]90:${\scriptstyle x_{p-3}}$}](x_{p-3}){} ++(0:1) node [label={[label distance=1pt]90:${\scriptstyle x_{p-2}}$}](x_{p-2}){} ++(0:1) node [label={[label distance=1pt]90:${\scriptstyle x_{p-1}}$}](x_{p-1}){} ++(0:1) node [label={[label distance=1pt]90:${\scriptstyle x_p}$}](x_p){} ++(0:1) node [label={[label distance=1pt]90:${\scriptstyle x_{p+1}}$}](x_{p+1}){} ++(0:1) node [label={[label distance=1pt]90:${\scriptstyle x_{p+2}}$}](x_{p+2}){}  ++(0:2) node [label={[label distance=1pt]90:${\scriptstyle x_{n-3}}$}](x_{n-3}){}++(0:1) node [label={[label distance=1pt]90:${\scriptstyle x_{n-2}}$}](x_{n-2}){} ++(0:1) node [label={[label distance=1pt]90:${\scriptstyle x_{n-1}}$}](x_{n-1}){} ++(0:1) node [label={[label distance=1pt]90:${\scriptstyle x_n}$}](x_n){} ;
\draw (x_1)++(-90:4) node [label={[label distance=1pt]-90:${\scriptstyle y_1}$}](y_1){} ++(0:1) node [label={[label distance=1pt]-90:${\scriptstyle y_2}$}](y_2){} ++(0:1) node [label={[label distance=1pt]-90:${\scriptstyle y_3}$}](y_3){} ++(0:1) node [label={[label distance=1pt]-90:${\scriptstyle y_4}$}](y_4){}++(0:2) node [label={[label distance=1pt]-90:${\scriptstyle y_{p-3}}$}](y_{p-3}){} ++(0:1) node [label={[label distance=1pt]-90:${\scriptstyle y_{p-2}}$}](y_{p-2}){} ++(0:1) node [label={[label distance=1pt]-90:${\scriptstyle y_{p-1}}$}](y_{p-1}){} ++(0:1) node [label={[label distance=1pt]-90:${\scriptstyle y_p}$}](y_p){} ++(0:1) node [label={[label distance=1pt]-90:${\scriptstyle y_{p+1}}$}](y_{p+1}){} ++(0:1) node [label={[label distance=1pt]-90:${\scriptstyle y_{p+2}}$}](y_{p+2}){}++(0:1) node [label={[label distance=1pt]-90:${\scriptstyle y_{p+3}}$}](y_{p+3}){}  ++(0:2) node [label={[label distance=1pt]-90:${\scriptstyle y_{n-2}}$}](y_{n-2}){} ++(0:1) node [label={[label distance=1pt]-90:${\scriptstyle y_{n-1}}$}](y_{n-1}){} ++(0:1) node [label={[label distance=1pt]-90:${\scriptstyle y_n}$}](y_n){};

\draw(y_p)++(-90:4) node [label={[label distance=1pt]-90:${\scriptstyle z}$}](z){} ;
\draw [->,very thick] (y_1)--(x_2);
\draw [->,very thick] (x_2)--(y_3);
\draw [->,very thick] (y_3)--(x_4);
\draw [->,very thick] [dashed](x_4)--(y_{p-3});
\draw [->,very thick](y_{p-3})--(x_{p-2});
\draw[->,very thick](x_{p-2})--(y_{p-1});
\draw[->,very thick] (y_{p-1})--(x_p);
\draw[->,very thick] (x_p)--(x_{p+1});
\draw[->,very thick] (x_{p+1})--(y_{p+2});
\draw[->,very thick][dashed] (y_{p+2})--(x_{n-3});
\draw[->,very thick] (x_{n-3})--(y_{n-2});
\draw[->,very thick] (y_{n-2})--(x_{n-1});
\draw[->,very thick] (x_{n-1})--(y_n);
\draw[->,very thick] (y_n)--(x_1);
\draw[->,very thick] (x_1)to[out=45,in=135](x_n);
\draw[->,very thick] (x_n)--(y_{n-1});
\draw[->,very thick] (y_{n-1})--(x_{n-2});
\draw[->,very thick][dashed] (x_{n-2})--(y_{p+3});
\draw[->,very thick] (y_{p+3})--(x_{p+2});
\draw[->,very thick] (x_{p+2})--(y_{p+1});
\draw[->,very thick] (y_{p+1})--(z);
\draw[->,very thick] (z)--(y_2);
\draw[->,very thick] (y_2)--(x_3);
\draw[->,very thick] (x_3)--(y_4);
\draw[->,very thick] [dashed](y_4)--(x_{p-3});
\draw[->,very thick] (x_{p-3})--(y_{p-2});
\draw[->,very thick] (y_{p-2})--(x_{p-1});
\draw[->,very thick] (x_{p-1})--(y_p);

\end{tikzpicture}
\caption{The $y_1$-$y_p$ path $P_1$ ($p$ is even).}
\end{figure}
\textbf{Case 6.2.} Let $y_1, y_p \in V(\mu (G))$ where $p$ is odd and $x_1 x_p \notin E(G)$. Now, assume $x_1$-$x_n$ Hamiltonian path $P$ in $G$, where
$P: x_1\ x_2\ x_3\ x_4\ldots x_{p-1}\ x_p \ x_{p+1}\ldots x_{n-2}\ x_{n-1}\ x_n$ and $x_1x_n \in E(G)$. Next, we construct the $y_1$-$y_p$ Hamiltonian path $P_1$ in $\mu (G)$ as follows ( Figure 6).\\
$$P_1: y_1 \  x_2 \  y_3 \  x_4 \ldots y_{p-2}\ x_{p-1}\ x_p\ y_{p+1}\ x_{p+2}\ldots y_{n-2}\ x_{n-1}\ y_n\ z\ y_{p-1}\ x_{p-2}\ y_{p-3}\ldots x_5\ y_4\ x_3\ y_2\ x_1\ x_n$$ \rightline{$\ y_{n-1}\ x_{n-2} \dots\ x_{p+3}\ y_{p+2}\ x_{p+1}\ y_p$.}\\
\begin{figure}[H]
\tikzstyle{every node}=[circle, draw, fill=black,
                        inner sep=0pt, minimum width=3pt]
\centering
\begin{tikzpicture}[scale=1.2]

\draw (0:0) node [label={[label distance=1pt]90:${\scriptstyle x_1}$}](x_1){} ++ (0:1)node [label={[label distance=1pt]90:${\scriptstyle x_2}$}](x_2){} ++ (0:1) node [label={[label distance=1pt]90:${\scriptstyle x_3}$}](x_3){} ++(0:1) node [label={[label distance=1pt]90:${\scriptstyle x_4}$}](x_4){} ++(0:2) node [label={[label distance=1pt]90:${\scriptstyle x_{p-2}}$}](x_{p-2}){} ++(0:1) node [label={[label distance=1pt]90:${\scriptstyle x_{p-1}}$}](x_{p-1}){} ++(0:1) node [label={[label distance=1pt]90:${\scriptstyle x_p}$}](x_p){} ++(0:1) node [label={[label distance=1pt]90:${\scriptstyle x_{p+1}}$}](x_{p+1}){} ++(0:1) node [label={[label distance=1pt]90:${\scriptstyle x_{p+2}}$}](x_{p+2}){} ++(0:2) node [label={[label distance=1pt]90:${\scriptstyle x_{n-2}}$}](x_{n-2}){} ++(0:1) node [label={[label distance=1pt]90:${\scriptstyle x_{n-1}}$}](x_{n-1}){} ++(0:1) node [label={[label distance=1pt]90:${\scriptstyle x_n}$}](x_n){} ;
\draw (x_1)++(-90:4) node [label={[label distance=1pt]-90:${\scriptstyle y_1}$}](y_1){} ++(0:1) node [label={[label distance=1pt]-90:${\scriptstyle y_2}$}](y_2){} ++(0:1) node [label={[label distance=1pt]-90:${\scriptstyle y_3}$}](y_3){} ++(0:1) node [label={[label distance=1pt]-90:${\scriptstyle y_4}$}](y_4){} ++(0:2) node [label={[label distance=1pt]-90:${\scriptstyle y_{p-2}}$}](y_{p-2}){} ++(0:1) node [label={[label distance=1pt]-90:${\scriptstyle y_{p-1}}$}](y_{p-1}){} ++(0:1) node [label={[label distance=1pt]-90:${\scriptstyle y_p}$}](y_p){} ++(0:1) node [label={[label distance=1pt]-90:${\scriptstyle y_{p+1}}$}](y_{p+1}){} ++(0:1) node [label={[label distance=1pt]-90:${\scriptstyle y_{p+2}}$}](y_{p+2}){} ++(0:2) node [label={[label distance=1pt]-90:${\scriptstyle y_{n-2}}$}](y_{n-2}){} ++(0:1) node [label={[label distance=1pt]-90:${\scriptstyle y_{n-1}}$}](y_{n-1}){} ++(0:1) node [label={[label distance=1pt]-90:${\scriptstyle y_n}$}](y_n){};

\draw(y_p)++(-90:3) node [label={[label distance=1pt]-90:${\scriptstyle z}$}](z){} ;
\draw [->,very thick] (y_1)--(x_2);
\draw [->,very thick] (x_2)--(y_3);
\draw [->,very thick] (y_3)--(x_4);
\draw [->,very thick][dashed] (x_4)--(y_{p-2});
\draw [->,very thick](y_{p-2})--(x_{p-1});
\draw [->,very thick] (x_{p-1})--(x_p);
\draw [->,very thick] (x_p)--(y_{p+1});
\draw [->,very thick] (y_{p+1})--(x_{p+2});
\draw [->,very thick][dashed] (x_{p+2})--(y_{n-2});
\draw [->,very thick] (y_{n-2})--(x_{n-1});
\draw [->,very thick] (x_{n-1})--(y_n);
\draw [->,very thick] (y_n)--(z);
\draw [->,very thick] (z)--(y_{p-1});
\draw [->,very thick] (y_{p-1})--(x_{p-2});
\draw [->,very thick][dashed] (x_{p-2})--(y_4);
\draw [->,very thick] (y_4)--(x_3);
\draw [->,very thick] (x_3)--(y_2);
\draw [->,very thick] (y_2)--(x_1);
\draw [->,very thick] (x_1)to[out=45,in=135](x_n);
\draw [->,very thick] (x_n)--(y_{n-1});
\draw [->,very thick] (y_{n-1})--(x_{n-2});
\draw [->,very thick] [dashed](x_{n-2})--(y_{p+2});
\draw [->,very thick] (y_{p+2})--(x_{p+1});
\draw [->,very thick] (x_{p+1})--(y_p);

\end{tikzpicture}
\caption{The $y_1$-$y_p$ path $P_1$ ($p$ is odd).}
\end{figure}
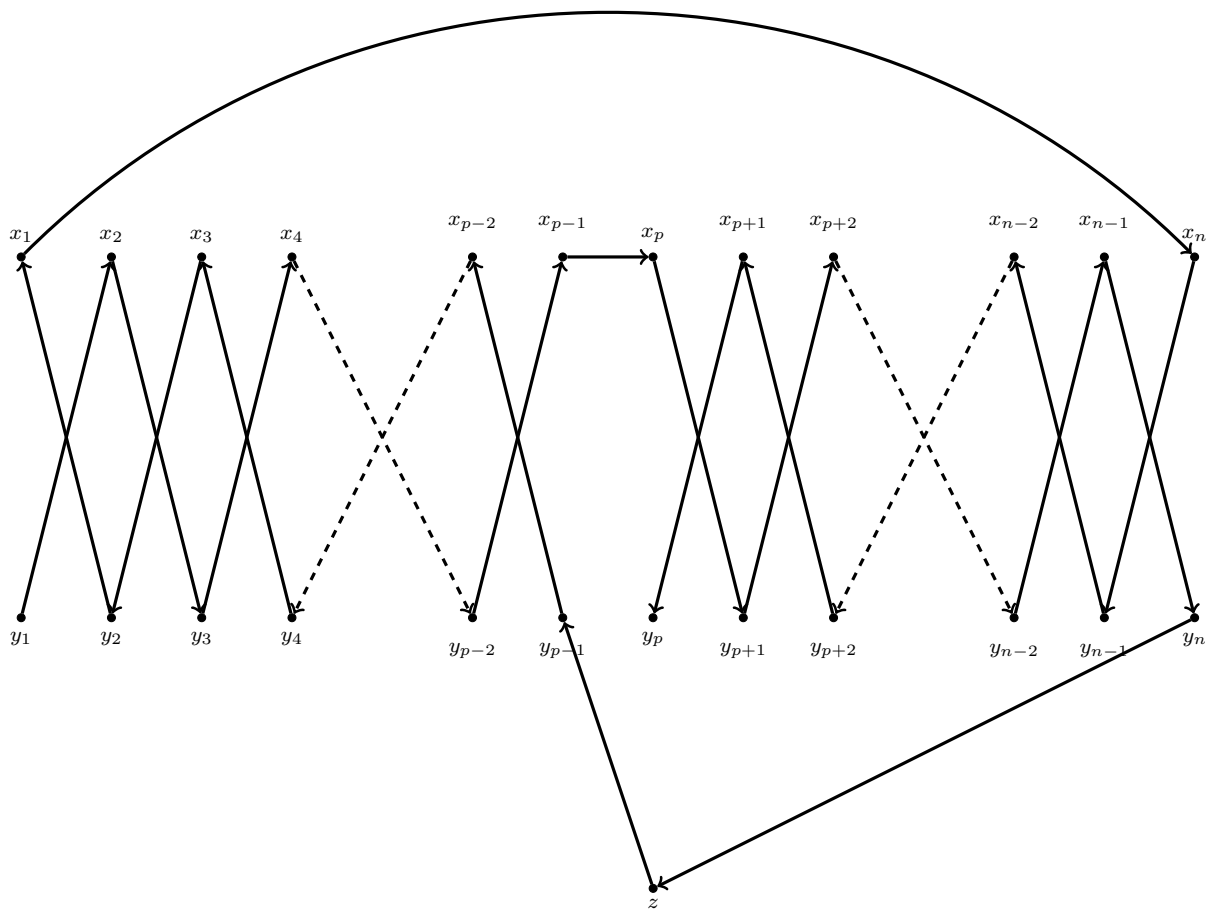

This completes the proof of Theorem 12.
\end{proof}


\begin{thebibliography}{99}
\bibitem{bc}J.A. Bondy and V. Chvátal, A method in graph theory. \emph{Discrete Math} 15(1976), 111--136.
\bibitem{ch}G. Chartrand, A.M. Hobbs, H.A. Jung, S.F. Kapoor and C.St.J.A. Nash-Williams, The square of a block is Hamiltonian-connected. \emph{J. Combin. Theory Ser. B} 16(1974), 290--292.
\bibitem{glp}G.Chartrand, L.Lesniak and P.Zang, Graphs and Digraphs, \emph{CRC Press} 2016.
\bibitem{c}V. Chvátal, On Hamiltonian ideas. \emph{J. Combin. Theory Ser. B} (1972), 163--168.
\bibitem{fm}D.C. Fisher, P.A. McKenna, and E.D. Boyer, Hamiltoniancity, diameter, domination, packing, and biclique partition of Mycielski's graphs. \emph{Disc. Appl. Math.} 84(1998), 93--105.
\bibitem{f}H. Fleischner, The square of every two-connected graph is Hamiltonian. \emph{J. Combin. Theory Ser. B} 16(1974), 29--34.
\bibitem{ts}T. Gangopadhyay and S.P. Rao Hebbare, r-Partite Self-Complementary Graphs-Diameters. \emph{Discrete Mathematics} 32(1980) 245-255.
\bibitem{g}R.J. Gould, Advances on Hamiltonian Problem - A survey, 2002, Emory University, Atlanta GA30322.
\bibitem{jm}W. Jarnicki, W. Myrvold, P. Saltzman, and S. Wagon, Properties, proved and conjectured, of Keller, Mycielski, and queen graphs. \emph{Ars Mathematica Contemporanea} 13(2017), 427--460.
\bibitem{or}O. Ore, Hamilton Connected graphs. \emph{J. Math Pures Appl.} 42(1963), 21--27.
\bibitem{se}M. Sekanina, On an ordering of the set of vertices of a connected graph. \emph{Spisy Priod Fak. Univ. Brno.} (1960), 137--141.
\bibitem{we}D.B. West, Introduction to graph theory (Second Editon). PHI Learning Pvt. Limited, New Delhi (2012).
\end{thebibliography}
\end{document}